\documentclass{article}
              
\usepackage[a4paper]{geometry}
\usepackage{amsmath,amsfonts,bm,amssymb}
\usepackage{amsthm}
\usepackage{float}
\usepackage{latexsym}
\usepackage{graphicx}
\usepackage{rotating}
\usepackage{caption}
\usepackage{epsfig}

\usepackage{subcaption}

\usepackage{color}

\usepackage{pdflscape}

\usepackage[normalem]{ulem}

  \usepackage[
    backend=biber,
    style=alphabetic,
  ]{biblatex}

\usepackage{lineno}

\newcommand{\be}{\begin{equation}}
\newcommand{\ee}{\end{equation}}
\newcommand{\ba}{\begin{eqnarray}}
\newcommand{\ea}{\end{eqnarray}}

\newtheorem{theorem}{Theorem}[section]

\theoremstyle{definition}

\theoremstyle{remark}

\usepackage{verbatim}

\newtheoremstyle{commenta}
  {6pt}
  {6pt}
  {\sffamily}
  {}
  {\sffamily \bfseries}
  {:}
  {.5em}
  {}

\theoremstyle{commenta}

\newcommand{\ev}{\mathrm{E}}
\newcommand{\var}{\mathrm{Var}}
\newcommand{\cov}{\mathrm{Cov}}

\newcommand{\beq}{\begin{equation}}     
\newcommand{\eeq}{\end{equation}}

\begin{document}

\title{Simple sign epistasis and evolutionary detours\\in fitness landscapes}
\author{
Paolo Ribeca$^{1,2}$, 
Alejandro Castro$^{3}$,
Alejandro Lage-Castellanos$^{3}$,\\
Alisa Sergeeva$^{4}$,
Sebastian Matuszewski$^{5}$,
Ruben Gustavo Paccosi$^{6}$,\\
Vitaly Belik$^{4}$,
Mahan Ghafari$^{7,10}$,
Joachim Krug$^{8}$, \\
Guillaume Achaz$^{9}$,
and Luca Ferretti$^{10}$\footnote{Email: luca.ferretti@gmail.com, luca.ferretti@ndm.ox.ac.uk}}

\date{}

\maketitle

{\footnotesize
\noindent%
(1) UK Health Security Agency, London, United Kingdom\\
(2) Biomathematics and Statistics Scotland, The James Hutton Institute, Edinburgh, United Kingdom\\
(3) Group of Complex Systems and Statistical Mechanics, Physics Faculty, University of Havana, Cuba\\
(4) System Modeling Group, Institute for Veterinary Epidemiology and Biostatistics, Freie Universitaet Berlin, Germany\\ 
(5) Accenture, Vienna, Austria\\ 
(6) National University of General Sarmiento, Buenos Aires, Argentine \\ 
(7) Department of Biology, University of Oxford, Oxford, United Kingdom\\  
(8) Institute for Biological Physics, University of Cologne, K\"oln, Germany\\  
(9) Stochastic Models for the Inference of Life Evolution (SMILE), CIRB, Collège de France, Université PSL, CNRS, INSERM, 75005 Paris, France \& Universit\'e Paris-Cit\'e, Paris, France\\ 
(10) Pandemic Sciences Institute and Big Data Institute, Nuffield Department of Medicine, University of Oxford, United Kingdom.
}




\begin{abstract}
In epistatic fitness landscapes, the fitness effect of a mutation depends on the genetic background and may even switch between deleterious and beneficial depending on the presence of another mutation. Epistatic interactions may cause both mutations to change the sign of each other’s fitness effects (reciprocal sign epistasis) or only one mutation to do so (simple sign epistasis). Both these forms of epistasis influence evolutionary trajectories. While reciprocal sign epistasis has been associated with multi-peaked landscapes and their ruggedness, the role and relative frequency of simple sign epistasis in fitness landscapes have not been systematically investigated. Here, we prove that the presence of simple sign epistasis is associated with evolutionary detours, i.e., indirect, longer fitness-increasing paths to fitness peaks that include back-mutations. We also show that in experimentally resolved, weakly epistatic landscapes, simple sign epistasis occurs much more frequently than reciprocal sign epistasis. This result is consistent with the theoretical predictions we derive for most landscape models, with the exception of the block model and of landscapes dominated by pairwise allelic incompatibilities, such as RNA stability landscapes. Our results suggest that detours represent a general feature of evolutionary trajectories in weakly epistatic landscapes.

\end{abstract}

\section{Introduction}
Fitness landscapes are a central concept of evolutionary theory, describing the relation between genotypes and their fitness \cite{Wright1932,svensson2012adaptive,srivastava2026}. These theoretical  objects provide the key conceptual framework to understand how populations evolve and how adaptive trajectories in genotype space are shaped by selection. However, the size of genotype spaces hinders the quantitative exploration of these landscapes at large scales. Several small fitness landscapes have been investigated and completely characterised in the last two-three decades through experimental evolution approaches \cite{de2014empirical}. These empirical studies, often based on microbial systems, have enabled the direct measurement of fitness values (or proxies for fitness) across all possible combinations of a limited number of mutations, offering partial insight into the topology of adaptive landscapes \cite{de2014empirical,bank2022epistasis}.
In recent years the size of such data sets has increased rapidly and we are now beginning to explore empirical fitness landscapes on functionally relevant scales \cite{Pokusaeva2019,Papkou2023,Westmann2024}.

A fundamental feature of fitness landscapes is epistasis, i.e. the dependence of the effect of mutations on their genetic background \cite{Phillips2008,Poelwijk2016,Domingo2019,bank2022epistasis}. Epistasis influences the structure of evolutionary paths, the accessibility of fitness peaks and the predictability of evolution \cite{de2014empirical,Szendro2013}. For this reason, summary statistics describing the amount, type and strength of epistasis are key for a quantitative understanding and classification of landscapes.

Epistasis in fitness landscapes can take different forms. A simple classification, based on the effect on the signs of the mutations, distinguishes between magnitude epistasis, simple sign epistasis and reciprocal sign epistasis \cite{weinreich2005perspective,Poelwijk2007,de2011causes}.
This classification considers a pair of mutations from a given genotype. If both mutations occur, and the fitness effect of any one of them does not depend on the order in which they occur, then we have magnitude epistasis (or no epistasis). In contrast, sign epistasis occurs when the fitness effect of one of these mutation changes sign depending on the presence of the other mutation. When this is true for just one of these two mutations, we denote it by \emph{simple sign epistasis} (also called single-sign epistasis in the literature). Reciprocal sign epistasis arises when both mutations switch between beneficial and deleterious depending on the presence/absence of the other mutation. Reciprocal sign epistasis has been widely studied, since its presence is related to the ruggedness and the number of local peaks in a fitness landscape \cite{Poelwijk2011,riehl2022occurrences,saona2022relation}. By contrast, the role of simple sign epistasis is not well understood, and the relative amount of simple vs reciprocal sign epistasis has not been systematically investigated, with few exceptions \cite{hwang2017genotypic,Kaznatcheev2019,srivastava2022incongruence}. This gap limits our understanding of how different forms of epistatic interactions contribute to evolutionary dynamics and constraints. In this paper, we shed some light on both these questions.

First, we describe the impact of simple sign epistasis on evolutionary trajectories. We prove that the presence of simple sign epistasis coincides with the presence of evolutionary detours in fitness-increasing trajectories, and we provide empirical evidence that the frequency of simple sign epistasis is correlated to the reduced accessibility of direct paths to peaks and to the length of evolutionary paths. This contrasts with the known role of RSE, which reduces accessibility of global optima by increasing the number of local fitness peaks.

Second, we obtain theoretical results for the fraction of motifs with simple sign and reciprocal sign epistasis in several landscape models and in empirical landscapes. We show that in all these models but two, sign epistasis tends to dominate over reciprocal sign epistasis when overall epistasis is weak, and it remains prevalent at intermediate levels of epistasis.

\section{Methods}
\subsection{Simple sign epistasis and reciprocal sign epistasis}

We consider genotypes as sequences of length $L$ with $A$ possible alleles in each position. A single mutation from a genotype is a substitution of an allele at a single position. The genotype space corresponds to the space of all possible sequences.  The fitness landscape is described by a function $f(g)$ on the genotype space, that gives the Malthusian fitness of each genotype (i.e. a real number).  For a mutation to allele $a$ at site $i$ from genotype $g$, we denote the resulting genotype as $g_{[i,a]}$. For biallelic landscapes, since there is no ambiguity on the final allele, we will use the simpler notation $g_{[i]}$. The fitness effect of a mutation from $g$ to $g_{[i,a]}$ is $\Delta_{i,a}f(g)=f(g_{[i,a]})-f(g)$.

We say that two mutations in different loci $i$ and $j$ interact epistatically if the fitness effect of the first mutation changes after the second mutation, i.e. $\Delta_{i,a}f(g)\neq \Delta_{i,a}f(g_{[j,a']})$. This is equivalent to a violation of linearity in the combination of fitnesses $f(g_{[j,a';i,a]})\neq f(g_{[i,a]})+f(g_{[j,a']})-f(g)$. There are many measures for the amount of epistasis in a landscape. The main measure that we will use in this paper is the correlation  $\gamma$ between the fitness effects of a mutation from two genotypes differing by a single mutation \cite{Ferretti2016}, defined as $\gamma=\mathrm{Cor}[\Delta_{[i,a]}f(g),\Delta_{[i,a]}f(g_{[j,a']})]$, which is 1 in the absence of epistasis and decreases for strong epistasis.

\begin{figure}[!htb]
\begin{center}
\includegraphics[width=13cm]{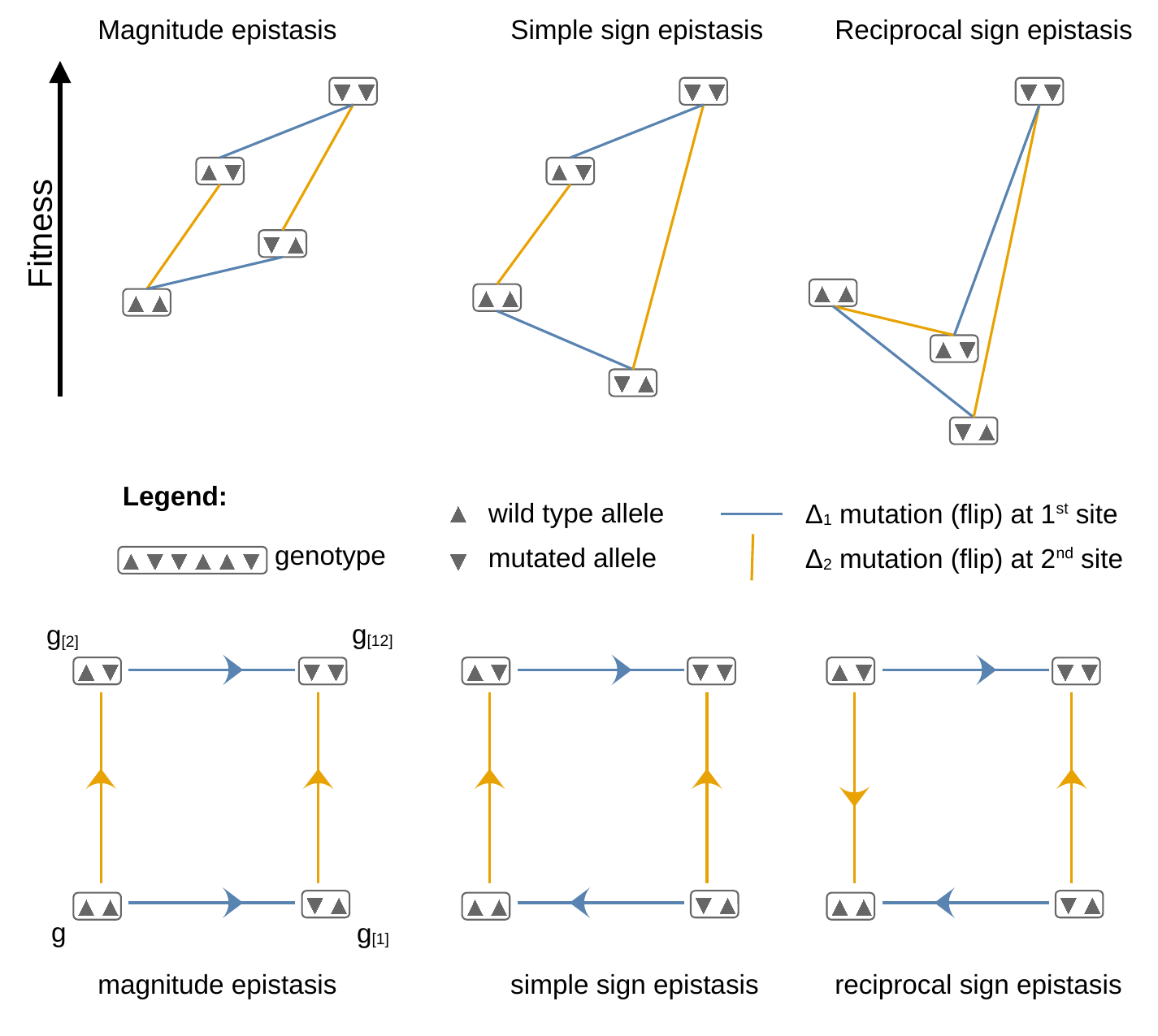}
\caption{\small \textbf{Above:} landscapes with 2 loci and 2 alleles per locus with different types of epistasis. \textbf{Below:} the corresponding fitness graphs. Arrows in the graph point the direction in of the higher fitness, not the sequence of mutations. Mutations go from genotype $g$  (bottom left) to $g_{[1,2]}$ (up right) either by applying allele flips $\Delta_1$ (blue, horizontal) and then $\Delta_2$ (yellow vertical), or in reverse order $\Delta_2$ first, then $\Delta_1$.}\label{fig:types}
\end{center}
\end{figure}

In some experimental landscapes, it could be difficult to quantify the precise fitness of a genotype, or its rank, or sometimes even the strength of fitness effects. The minimal information available is the sign of the fitness effect, i.e., if the mutation is beneficial or deleterious (or neutral). This information is enough to define a \emph{fitness graph}, i.e., an oriented graph with arrows pointing in the direction of fitness increase \cite{Crona2013}.  Examples of these graphs can be seen in Figure \ref{fig:types}.

From the fitness graph, it is possible to classify three different types of epistasis for each pair of mutations in different sites from a given genotype. Imagine to mutate one site after the other. If the signs of the effects of the mutations do not change with their order, the pair is either non-epistatic or, more generally, they show \emph{magnitude epistasis} (ME); otherwise, they have sign epistasis. If one of them changes sign but the other does not, they show \emph{simple sign epistasis} or single-sign epistasis (SSE). If both change sign, then they show \emph{reciprocal sign epistasis} (RSE). This classification is illustrated in Figure \ref{fig:types}. 

It is also possible to measure the frequency of different types of epistasis, i.e. the fraction of pairs of mutations showing magnitude, sign or reciprocal sign epistasis. We denote these fractions as $\phi_m$, $\phi_{ss}$ and $\phi_{rs}$. These measures will be the main focus of this paper. The combinations $\phi_{ss}+\phi_{rs}$ (i.e. overall sign epistasis) or $\phi_{ss}+2\phi_{rs}=1-\gamma^*$ can be used as measures of overall epistasis in the fitness graph. The combination $1-\gamma^*=\phi_{ss}+2\phi_{rs}$ has a simple interpretation: it is the probability that the effect of a random mutation changes sign after another random mutation occurs in another locus. This quantity is therefore related to the change in sign of fitness effects across different genetic backgrounds; in fact, $\gamma^*=1-\phi_{ss}-2\phi_{rs}=\mathrm{Cor}[\mathrm{sign}(\Delta_{[i,a]}f(g)),\mathrm{sign}(\Delta_{[i,a]}f(g_{[j,a']}))]$ is the equivalent of the correlation of fitness effects $\gamma$ for fitness graphs \cite{Ferretti2016}. Note that these two measures can differ significantly, since $\gamma$ accounts for the magnitude of changes in fitness effects (and therefore accounts for magnitude epistasis as well) while $\gamma^*$ considers only changes in their sign (i.e. changes in the fitness graph), irrespective of the actual magnitude of changes in fitness effects. 

In this paper, if one of the mutations in a square motif is neutral, we classify the motif as showing magnitude epistasis in order not to artificially inflate the number of epistatic interactions.

\subsection{Frequency of sign epistasis in landscape models}

The fractions of magnitude, sign and reciprocal sign epistasis $\phi_m,\phi_{ss},\phi_{rs}$ 
sum to $\phi_m+\phi_{ss}+\phi_{rs}= 1$ but the constraints on their relative frequency are unclear. In fact, some empirical landscapes contain only motifs with magnitude epistasis, and it is easy to design theoretical biallelic landscapes with reciprocal sign epistasis only (``eggbox'' landscapes \cite{Ferretti2016}). 
In this section, we compute the amount of sign and reciprocal sign epistasis $\phi_{ss},\phi_{rs}$ in several landscape models. For simplicity, we consider only models where all mutations are beneficial or deleterious, but not neutral.

\subsubsection{The House of Cards model}
This is a strongly epistatic random model of fitness. Each genotype is assigned a fitness that is independent from the others, but identically distributed. It corresponds to the most epistatic scenario for the next models (RMF with $s=0$, NK with $K=L-1$). Irrespective of the distribution, the expected fractions of SSE and RSE are
\begin{align}
\phi_{rs}&=\frac{1}{3} \\
\phi_{ss}&=\frac{1}{3}.
\end{align}

\subsubsection{Unstructured Gaussian models}
We present here a large class of models for which we can compute the amount of sign and reciprocal sign epistasis in a general way. These models correspond to the so-called isotropic Gaussian random field models in the literature \cite{stadler1999random,zhou2022higher}. They share the following features: (i) the fitnesses of all genotypes are arbitrarily correlated centered Gaussian variables; (ii) they have no structure, i.e. their distribution is invariant under arbitrary exchange of loci or alleles within a locus. 

\begin{figure}
\begin{center}
\includegraphics[width=13cm]{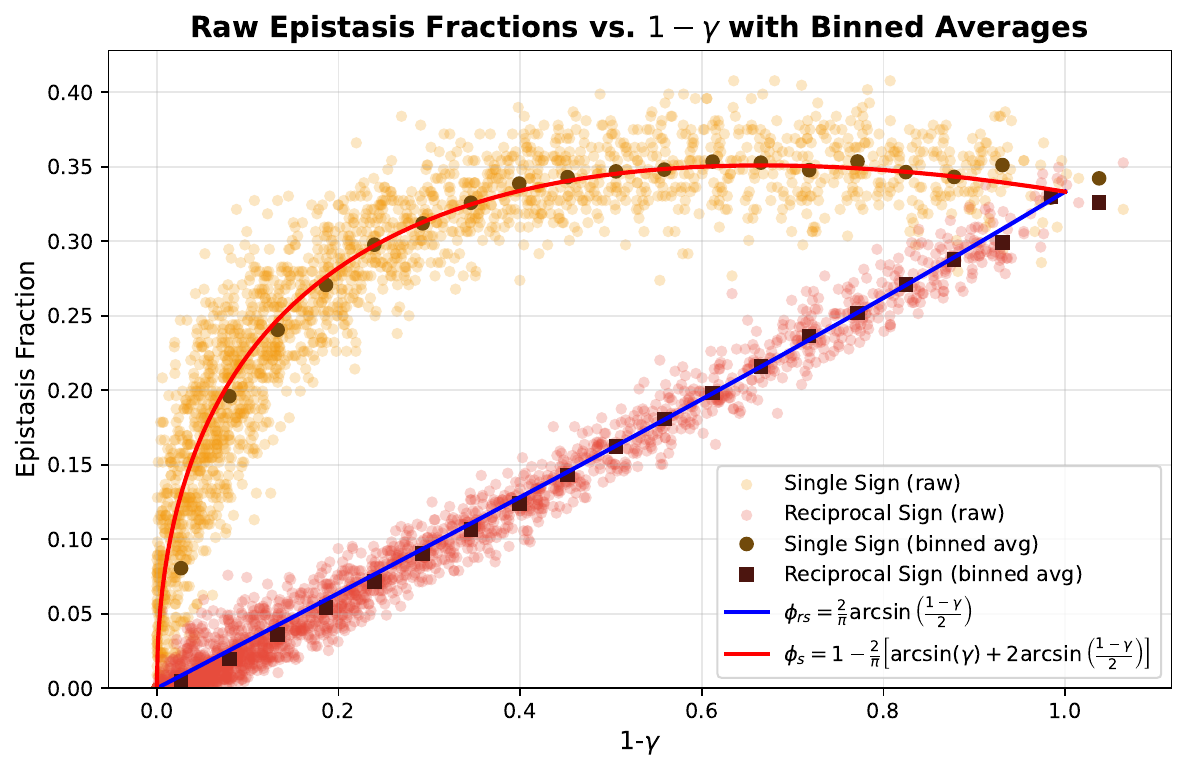}
\caption{\small Comparison between theory and simulations for the fractions $\phi_{ss},\phi_{rs}$ of SSE and RSE motifs for a variant of the RMF model that is also an unstructured Gaussian model. Light points represent 100 realizations of the biallelic model with $L=7$, using brute-force exhaustive counting of epistatic occurrences. Dark points are their averages. The solid lines correspond to the analytical predictions from unstructured Gaussian models, i.e. eqs. (\ref{relfrs}) and (\ref{relfs}), or equivalently (\ref{unstructrmf_rs}) and (\ref{unstructrmf_s}). }\label{fig:counting_vs_analytical}
\end{center}
\end{figure}

\begin{figure}
\begin{center}
\includegraphics[width=13cm]{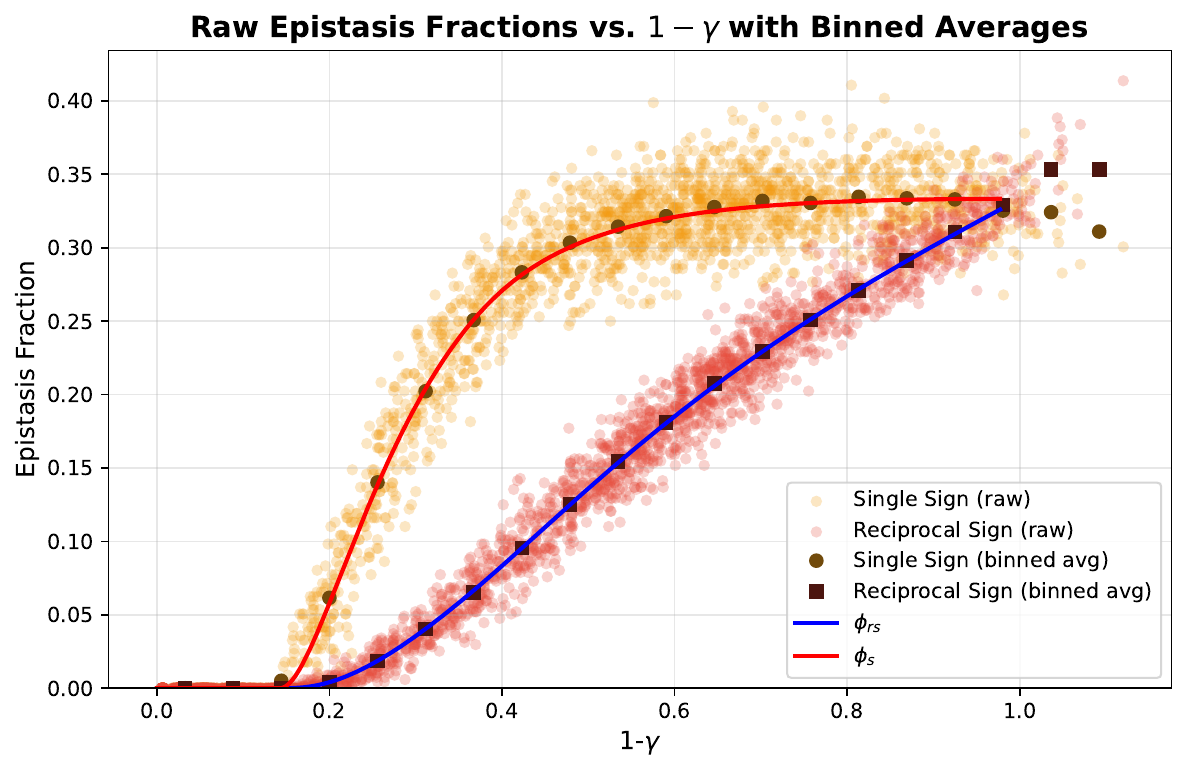}
\caption{\small Comparison between theory and simulations for the fractions $\phi_{ss},\phi_{rs}$ of SSE and RSE motifs for the classical RMF model with a uniformly distributed random component. Light points represent 100 realizations of the biallelic model with $L=7$, using brute-force exhaustive counting of epistatic occurrences. Dark points are their averages. The solid lines correspond to the analytical predictions for this model, i.e. eqs. (\ref{unifrmf_rs}) and (\ref{unifrmf_s}). }\label{fig:RMF_counting_vs_analytical}
\end{center}
\end{figure}

This class includes some of the models discussed above (Gaussian-distributed House of Cards) and below, as well as a plethora of models with an arbitrary spectrum of pairwise and higher-order epistatic interactions.

For all these models, the amount of sign and reciprocal sign epistasis depends only on the epistasis of the landscape as measured by the correlation of fitness effects $\gamma$ \cite{Ferretti2016}:
\begin{align}
\phi_{rs}&=\frac{2}{\pi}\arcsin\left(\frac{1-\gamma}{2}\right) \label{relfrs}\\
\phi_{ss}&= 1-\frac{2}{\pi}\left[\arcsin\left(\gamma\right)+2\arcsin\left(\frac{1-\gamma}{2}\right)\right].\label{relfs}
\end{align}
Derivation is shown in Supplementary Information. A comparison with simulated results for a specific unstructured model discussed in the next section is presented in Figure \ref{fig:counting_vs_analytical}.

\subsubsection{The Rough Mount Fuji model}

We consider two versions of this model. 

The classical one has a fixed additive fitness effect $s$ for each locus, plus an ``House of Cards'' contribution which adds an independent random value to the fitness of each genotype \cite{neidhart2014adaptation}. For a uniform distribution of random effects between 0 and $\sigma$, we derive the expected values (see Supplementary for derivations and Figure \ref{fig:RMF_counting_vs_analytical} for comparison with simulated results):
\begin{align}
\phi_{rs}&=\begin{cases} 0 & \mathrm{for}\ \frac{s}{\sigma}\geq {1}\\
\frac{2}{3}\left(1-\frac{s}{\sigma}\right)^3
& \mathrm{for}\ 1>\frac{s}{\sigma}>\frac{1}{2}\\ 
\frac{1}{3}-2\left(\frac{s}{\sigma}\right)^2\left(1-\frac{s}{\sigma}\right)
& \mathrm{for}\ \frac{s}{\sigma}\leq \frac{1}{2}\end{cases} \label{unifrmf_rs}\\
\phi_{ss}&=\begin{cases} 0 & \mathrm{for}\ \frac{s}{\sigma}\geq {1}\\
\left(1-\frac{s}{\sigma}\right)^2\left(-\frac{1}{3}+\frac{10}{3}\frac{s}{\sigma}-\left(\frac{s}{\sigma}\right)^2\right)
& \mathrm{for}\ 1>\frac{s}{\sigma}>\frac{1}{2}\\ 
\frac{1}{3}- \left(\frac{s}{\sigma}\right)^4
& \mathrm{for}\ \frac{s}{\sigma}\leq \frac{1}{2}.\end{cases} \label{unifrmf_s}
\end{align}

We also compute the epistasis in a second version, where both the additive contribution for each locus and the random ``House of Cards'' component are centered Gaussian distributed random variables, with variances $\sigma_s^2$ and $\sigma_{HoC}^2$ respectively. This is also an unstructured Gaussian model. The expected values for this model are
\begin{align}
\phi_{rs}&=\frac{2}{\pi}\arcsin\left(\frac{\sigma_{HoC}^2}{\sigma_{s}^2+2\sigma_{HoC}^2}\right) \label{unstructrmf_rs}\\
\phi_{ss}&=1-\frac{2}{\pi}\left[\arcsin\left(\frac{\sigma_{s}^2}{\sigma_{s}^2+2\sigma_{HoC}^2}\right)+2\arcsin\left(\frac{\sigma_{HoC}^2}{\sigma_{s}^2+2\sigma_{HoC}^2}\right)\right].\label{unstructrmf_s}
\end{align}

\subsubsection{The NK model}
The NK model describes the fitness of a genotype as a sum of independent contributions for each locus \cite{Kauffman1989}. The contribution of each locus is affected by random ``House of Cards'' interactions with $K$ other loci (here assumed to be Gaussian-distributed). Different models can be obtained by different choices of interacting loci \cite{hwang2018universality}. Although it is often assumed that this choice would have only a small effect on the model and especially on the number of local maxima \cite{weinberger1991local}, this has been proven not to be the case \cite{hwang2018universality}. In fact, we will see that the choice of interacting neighbours can have a dramatic impact on $\phi_{ss}$ and $\phi_{rs}$.

The first model (random NK) assumes a random choice of interacting loci. While exact analytical results are difficult in this case, an approximate result can be derived by a mean-field approach for the expected values \cite{hwang2018universality}. The mean-field NK model is an unstructured Gaussian model, therefore:
\begin{align}
\phi_{rs}&\approx \frac{2}{\pi}\arcsin\left(\frac{K}{2(L-1)}\right) \\
\phi_{ss}&\approx 1-\frac{2}{\pi}\left[\arcsin\left(1-\frac{K}{L-1}\right)+2\arcsin\left(\frac{K}{2(L-1)}\right)\right].
\end{align}

The second model (adjacent NK) assumes that interactions involve adjacent loci, i.e. the $K$ closest neighbours. The expected values are:
\begin{align}
\phi_{rs}&= \frac{2}{\pi (L-1)}\sum_{d=1}^{\lfloor L/2 \rfloor }\frac{2}{1+\delta_{d,L/2}} \arcsin\left(\frac{K+1-\min(d,K+1,L-K-1)}{2(K+1)}\right) \\
\phi_{ss}&=1-\frac{2}{\pi(L-1)}\sum_{d=1}^{\lfloor L/2 \rfloor } \frac{2}{1+\delta_{d,L/2}}\left[\arcsin\left(\frac{\min(d,K+1,L-K-1)}{K+1}\right)+\right.\nonumber\\
&\left.+2\arcsin\left(\frac{K+1-\min(d,K+1,L-K-1)}{2(K+1)}\right)\right].
\end{align}

\subsubsection{The block model}

This model belongs to the family of NK models \cite{perelson1995protein,Schmiegelt2014}. It assumes that the sequence is divided into approximately $L/(K+1)$ blocks, each locus interacting with all loci in the same block, resulting in a ``House of Cards'' model within each block. In this block model, the types of epistasis have the expected frequencies (irrespective of the distribution):
\begin{align}
\phi_{rs}&=\frac{K}{3(L-1)} \\
\phi_{ss}&=\frac{K}{3(L-1)}.
\end{align}

\subsection{Statistical analysis of empirical landscapes}
We analyse a collection of experimentally characterised landscapes \cite{Chou2011,Khan2011,Tan2011,Weinreich2006,daSilva2010,deVisser1997} already presented and discussed in \cite{ferretti2018evolutionary}. To obtain a more complete dataset, for each landscape we consider all its complete sublandscapes obtained by mutating only a subset of 3, 4 or 5 loci. All statistical analyses on these landscapes were performed using MAGELLAN \cite{Magellan}. These analyses do not incorporate experimental noise in fitness measurements.

\section{Results}
\subsection{Types of epistasis and their impact on evolution}
\subsubsection{Theoretical results}

Simple sign epistasis (SSE) alone does not imply strong evolutionary constraints. In fact, while RSE is related to the presence of multiple local peaks in the landscape, landscapes with SSE but no RSE have a single peak that is always accessible from any genotype, as in non-epistatic landscapes. However, SSE changes the structure of evolutionary trajectories, increasing their length - so that the peak is always reachable, but only through longer ``evolutionary detours'' involving back-mutations. 

We define a \emph{direct accessible path} between two genotypes as a fitness-increasing trajectory that points directly towards the final genotype, i.e. the Hamming distance from the final genotype decreases by 1 at every step \cite{krug2024evolutionary}. Any other fitness-increasing path should contain side-mutations (multiple mutations in the same site) or back-mutations (mutations that happen in both directions along the path) and it is an indirect path. An indirect path that involves a back-mutation is a \emph{reversing path}. Reversing accessible paths are the clearest examples of evolutionary detours. 

In the Strong Selection/Weak Mutation (SSWM) approximation,  populations can evolve only along fitness-increasing paths and the average time length of the trajectory is proportional to the number of mutational steps. Hence, indirect paths increase the number of substitutions and slow down the rate of fitness increase.

The precise role of SSE is actually more subtle.  In fact, SSE has a double effect. Without sign epistasis, non-neighbour genotypes along a fitness-increasing path are connected by $d!>1$ direct paths if they differ at $d$ loci. Compared to non-epistatic landscapes, SSE reduces this number of direct paths between genotypes:
\begin{theorem}
The presence of SSE is a necessary and sufficient condition for the existence of pairs of non-neighbour genotypes connected by exactly one direct accessible path.
\end{theorem}
\noindent This means that SSE is the minimal local motif that reduces the number of direct accessible paths, breaking the redundancy across these direct uphill paths.

At the same time, SSE generates new (longer) indirect paths containing back-mutations, as it is clear from the following:
\begin{theorem}
The presence of SSE is a necessary and sufficient condition for the presence of reversing paths in the landscape.
\end{theorem}
\noindent This means that SSE is exactly the local signature that allows detours with back-mutations.

Looking at the fitness graph for two mutations with SSE interactions (bottom of Figure \ref{fig:types}), it is easy to understand the idea behind these statements. In fact, in this simple case, the genotype farther from the peak has only one direct accessible path to the peak, instead of two as in the non-epistatic case. However, another genotype gained a reversing accessible path of length three to the peak, that would not exist in non-epistatic landscapes.

On the other hand, reciprocal sign epistasis (RSE) is well recognised as key feature of fitness landscapes. Its importance stems from a celebrated theorem \cite{Poelwijk2011} that connects it to multiple peaks in the landscape: 
\begin{theorem}[Poelwijk 2011; Saona et al. 2022; Riehl et al. 2022] 
Multiple fitness peaks exist only in the presence of RSE. More precisely, in a landscape with $K \geq 2$ peaks there are at least $K-1$ RSE motifs.
\end{theorem}


In general terms, RSE is related to evolutionary constraints in reaching a peak of the landscape, as can be appreciated by this more general corollary: 
\begin{theorem} If there is at least one genotype that has no direct fitness-increasing path to a peak in the landscape, then there is RSE.
\end{theorem}


Hence, RSE can disrupt direct paths between genotypes, making it impossible to move among different regions of the landscape through a sequence of beneficial mutations. This can be intuitively understood also from the fitness graph of two mutations interacting via RSE (bottom of Figure \ref{fig:types}). In this simple case, each of the two peaks has no direct fitness-increasing path - and actually no path at all - to the other peak.

More generally, these properties suggests that in landscapes with ME and SSE only, the minimum number of mutations to reach the peak is the same as in non-epistatic landscapes, since at least one direct path to the peak must be accessible; however, the expected time to reach the peak is longer due to the possibility of evolutionary detours \cite{Kaznatcheev2019}.

\subsubsection{Role of epistasis in experimental landscapes}

To confirm the general validity of the theoretical insights above, we explore the empirical relationship between SSE, RSE and accessible paths in experimentally characterised landscapes, and more specifically in a collection of such landscapes \cite{Chou2011,Khan2011,Tan2011,Weinreich2006,daSilva2010,deVisser1997} and their sublandscapes. All these landscapes are biallelic, therefore all indirect fitness-increasing paths in these landscapes are reversing paths. In Table \ref{table_accessibility}, we build a linear model based on $\phi_{ss}$ and $\phi_{rs}$ for each of the measures of accessibility listed in the table, and we report the inferred coefficients for both statistics. 

Both SSE and RSE reduce the number of accessible direct paths as expected, with a more marked effect for RSE. RSE also reduces all other statistics of accessibility. Instead, as expected, SSE has a positive effect on the number of indirect accessible paths (i.e. all accessible paths minus the direct ones), on their length, on as well as the overall number and length of paths per peak. Note that the increase in indirect paths due to SSE does not compensate for the reduction in direct paths. Hence, while not necessarily precluding access to the highest peak, SSE acts anyway as a constraint on evolutionary trajectories. 

Notably, SSE increases the predictability of evolution by reducing the number of direct accessible paths while still leaving some accessible paths, which are therefore more likely to be taken.

\begin{table}
\begin{center}{\small
\begin{tabular}{|c|c|c|c|c|} \hline
Landscapes & \multicolumn{2}{c|}{ 3 loci } & \multicolumn{2}{c|}{ 4 loci } \\ \hline
Linear coefficients of: & SSE ($\phi_{ss}$) & RSE ($\phi_{rs}$) & SSE ($\phi_{ss}$) & RSE ($\phi_{rs}$)  \\  \hline
\# of direct accessible paths to global optimum        &$-3.3^{***}$&$-6.7^{***}$&$-16.4^{**}$&$-26.8^{***}$\\
\# of accessible paths to global optimum       &$-2.6^{***}$&$-7.2^{***}$&$-11.4^{}$&$-36.2^{***}$\\
Mean length of accessible paths         &$0.4^{***}$&$-0.3^{***}$&$1.2^{***}$&$-1.8^{***}$\\ 
\# of paths per peak &$1.6^{**}$&$-2.4^{***}$&$13.0^{}$&$-40.4^{***}$\\
Total length of paths per peak  &$1.3^{***}$&$-2.6^{***}$&$6.7^{*}$&$-17.2^{***}$\\ \hline
\end{tabular}
\caption{Empirical contributions of the fraction of SSE and RSE to several measures of accessibility from fitness-increasing paths, defined in terms of coefficients of a linear model estimated across empirical biallelic landscapes with 3 and 4 loci. All coefficients are unnormalised, but the coefficients for SSE and RSE can be directly compared. Both RSE and (to a lesser extent) SSE contribute negatively to statistics related to the number of accessible paths to the highest fitness peaks; however, SSE contributes positively to statistics related to the length of accessible paths, while RSE contributes negatively. (Significance: $^*$ $p<0.05$; $^{**}$ $p<0.01$; $^{***}$ $p<0.001$.)}\label{table_accessibility}
}\end{center}
\end{table}

\subsection{Quantifying sign epistasis in experimental and model landscapes}

The amount of sign epistasis depends on the structure and nature of the landscape. In fact, there are epistatic landscapes with little sign epistasis. For example, a non-linear fitness dependence on weakly epistatic phenotypes 
(as in global epistasis \cite{diaz2023global}) could generate a significant amount of epistasis without any sign epistasis. However, it is equally or more interesting to look at the relative contributions of SSE and RSE. 

First, we look at the fraction of RSE and SSE in the collection of experimentally resolved sublandscapes of different size described above. The relative contribution of RSE to sign epistasis is small, especially for weak levels of epistasis, as shown in Figure \ref{figs_se_exp}(a). Even at strong levels of epistasis, its contribution tends to be at most comparable to SSE. The same trend is present across different landscape sizes, as shown in Figures \ref{figs_se_exp}(b,c,d). 

Notably, the average amounts of SSE and RSE in these empirical landscapes are well predicted by the expected SSE and RSE for unstructured models. This is the case for landscape of any size $L=3,4,5$ as illustrated in Figures \ref{figs_se_exp}(b,c,d).

\begin{figure}
    \centering
    \begin{subfigure}[b]{0.47\textwidth}
        \caption{RSE versus sign epistasis}
        \vspace{-10pt}
	\includegraphics[width=\textwidth]{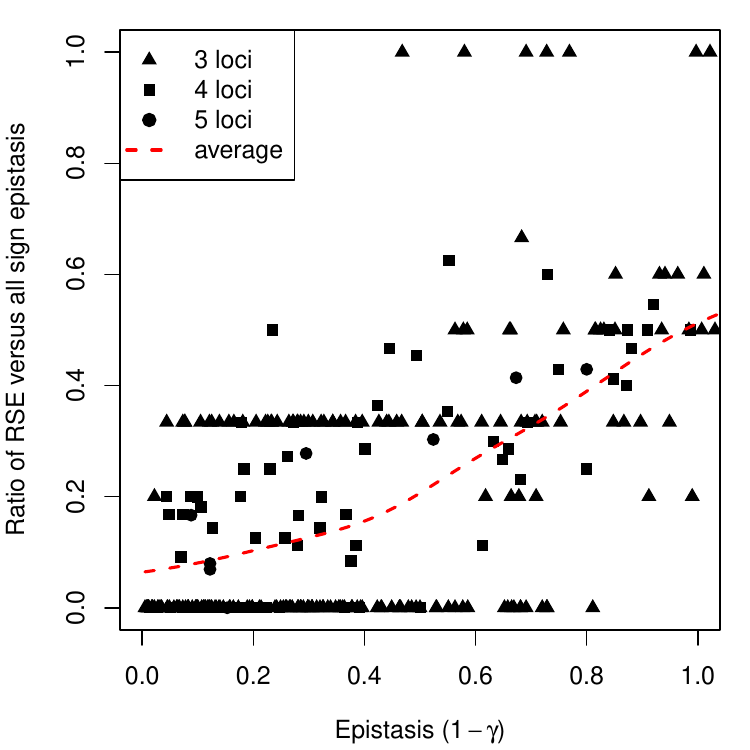}
        \label{fig:exp}
    \end{subfigure}
    ~ 
    \begin{subfigure}[b]{0.47\textwidth}
        \caption{SSE and RSE for 3 loci}
        \vspace{-10pt}
	\includegraphics[width=\textwidth]{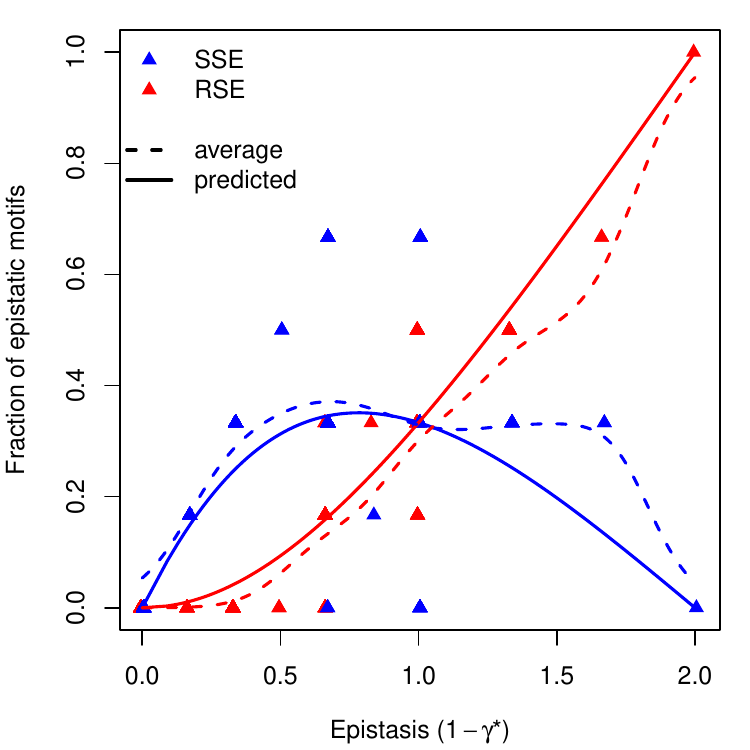}
        \label{fig:fit3loci}
    \end{subfigure} \\
    \begin{subfigure}[b]{0.47\textwidth}
        \caption{SSE and RSE for 4 loci}
        \vspace{-10pt}
	\includegraphics[width=\textwidth]{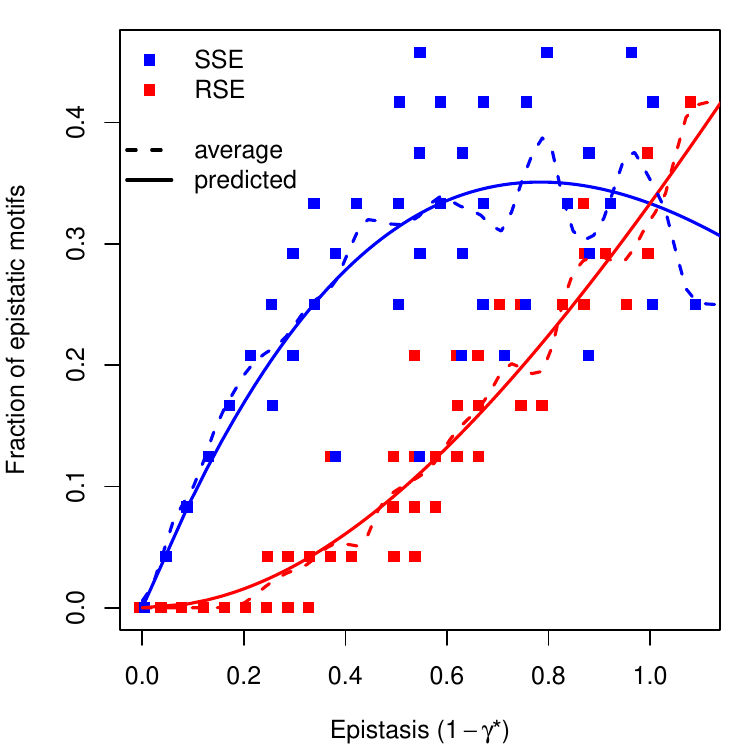}
        \label{fig:fit4loci}
    \end{subfigure}
    ~ 
    \begin{subfigure}[b]{0.47\textwidth}
        \caption{SSE and RSE for 5 loci}
        \vspace{-10pt}
	\includegraphics[width=\textwidth]{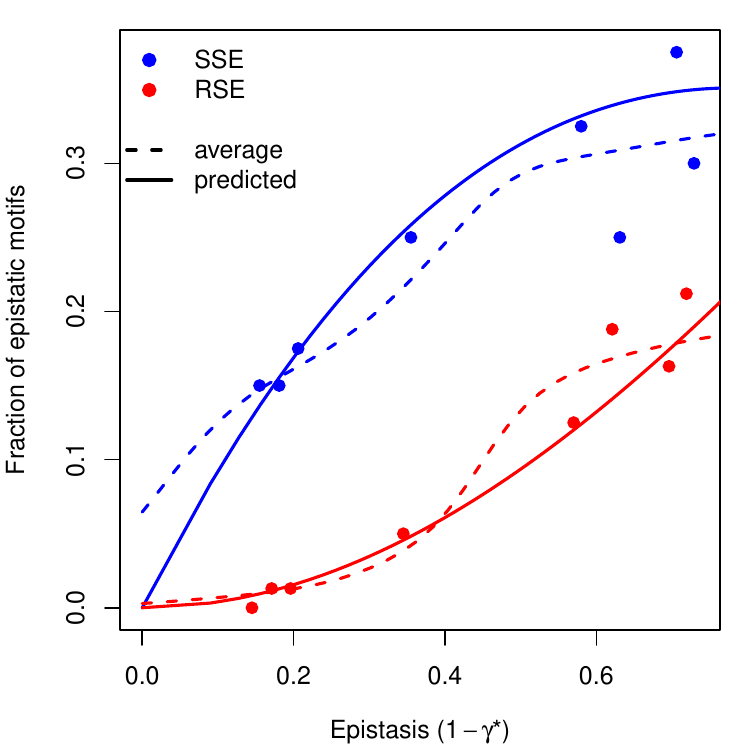}
        \label{fig:fit5loci}
    \end{subfigure}
    \caption{\small \textbf{(a)} Empirical fraction of RSE motifs versus all motifs involving sign epistasis, i.e. $\frac{\phi_{rs}}{\phi_{ss}+\phi_{rs}}$, computed for empirical (sub)landscapes of different size $L=3,4,5$ as a function of the estimated epistasis $1-\gamma$; the moving average of $\frac{\phi_{rs}}{\phi_{ss}+\phi_{rs}}$ is shown by the dashed line.
    \textbf{(b)} Comparison of SSE and RSE for empirical landscapes with theoretical predictions for unstructured Gaussian landscapes. The points represent the fractions $\phi_{ss},\phi_{rs}$ of motifs with SSE and RSE as a function of fitness graph epistasis $1-\gamma^*$ in empirical biallelic (sub)landscapes with 3 loci, with the moving averages represented by dashed lines; the continuous lines indicate the expected values for landscapes with unstructured Gaussian interactions.
    \textbf{(c)} Comparison of SSE and RSE for empirical 4-loci landscapes with theoretical predictions for unstructured Gaussian landscapes.
    \textbf{(d)} Comparison of SSE and RSE for empirical 5-loci landscapes with theoretical predictions for unstructured Gaussian landscapes.
    }\label{figs_se_exp}
\end{figure}


The expected amount of SSE and RSE, as well as the total amount of sign epistasis, is illustrated as a function of the local epistasis $1-\gamma$ in Figure \ref{fig_models_gamma} for a variety of landscape models. All unstructured models, including Rough Mt Fuji and NK variants, show the exact same behaviour. The NK model with interactions between adjacent loci and the RMF model with uniform distribution show slightly different behaviour; however, also these models show a monotonic relation between epistasis and sign epistasis, as well as a prevalence of SSE over RSE for all reasonable values of epistasis.
The differences between different models become less pronounced when exploring the amount of SSE and RSE versus the fitness graph epistasis $1-\gamma^*$ in 
Figure \ref{fig_models}. This Figure measures the relative contribution of each type of epistatic motifs, rather than the absolute one.

The exception to the overall trend is provided by the block model. This model has equal amount of SSE and RSE for any amount of epistasis, and it is therefore a clear outlier in Figure \ref{fig_models}. However, note that the amount of RSE is not far from other models with comparable epistasis, and what differentiates it is the much lower amount of sign epistasis (Figure \ref{fig_models_gamma}).

\begin{figure}
\begin{center}
\includegraphics[width=0.8\textwidth]{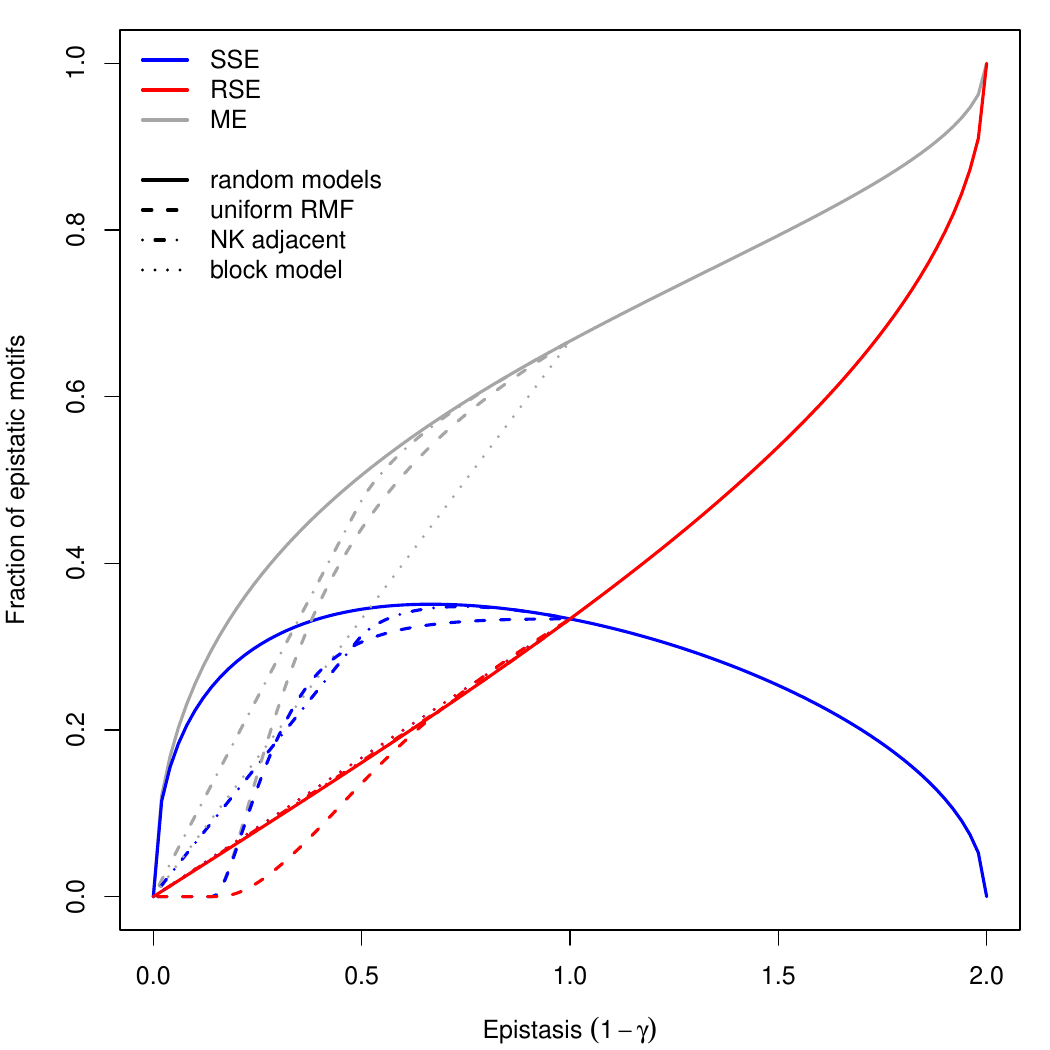}
\caption{\small Theoretical predictions for the fractions $\phi_{ss}$, $\phi_{rs}$ as well as the overall sign epistasis $\phi_{ss}+\phi_{rs}$ versus local epistasis $1-\gamma$ for the models discussed in this paper. The continuous lines represent the predictions from unstructured random (Gaussian) models, i.e. the same lines as in Figure \ref{figs_se_exp}.}\label{fig_models_gamma}
\end{center}\end{figure}
\begin{figure}
\begin{center}
\includegraphics[width=0.8\textwidth]{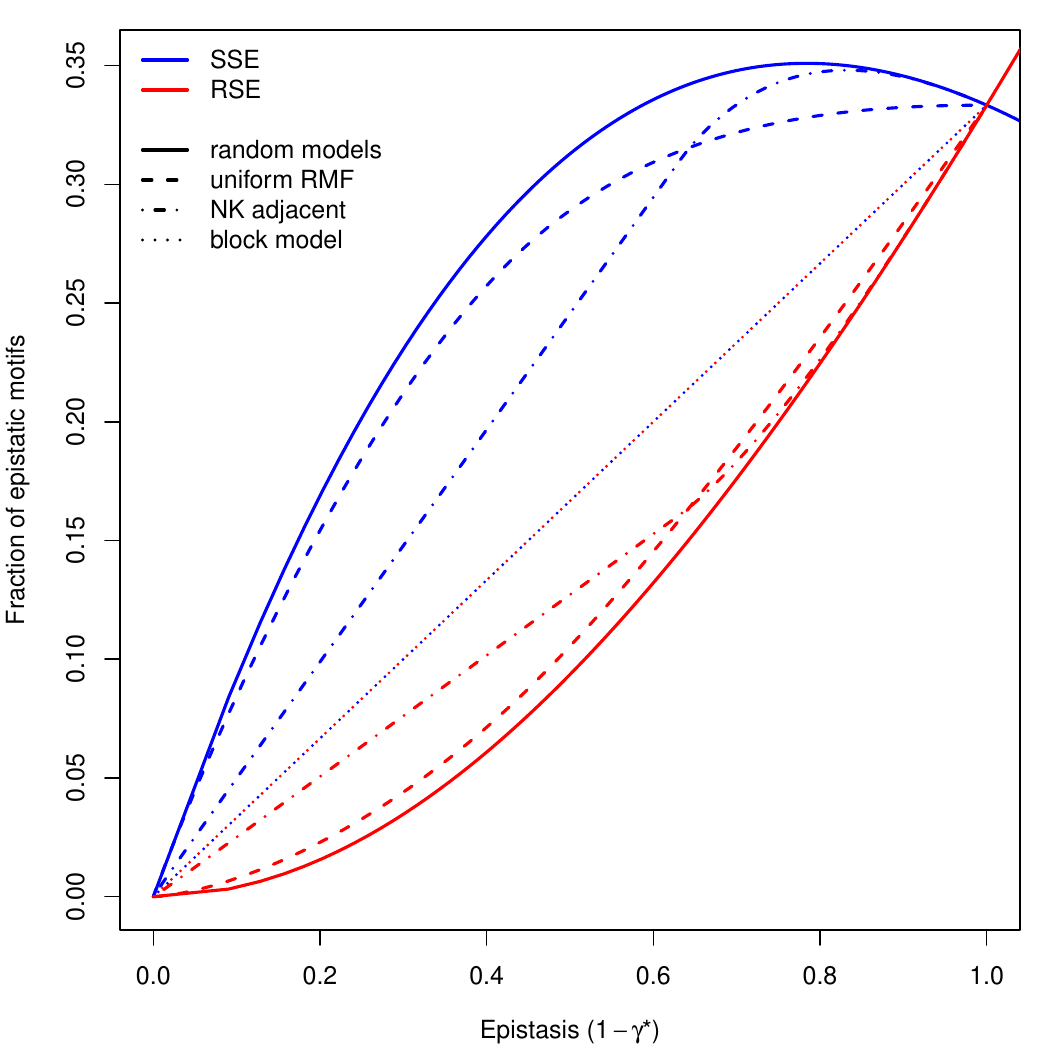}
\caption{\small Theoretical predictions for the fractions $\phi_{ss}$, $\phi_{rs}$ versus ``fitness graph'' epistasis $1-\gamma^*$ for the models discussed in this paper.}\label{fig_models}
\end{center}\end{figure}

\subsubsection{Allelic Incompatibilities models}
None of the previous models showed a prevalence of RSE over SSE for weakly epistatic landscapes. This is a feature of most existing landscape models. However, it is not difficult to think of a model dominated by reciprocal sign epistasis, by introducing strong pairwise compensatory interactions. We discuss the simplest models in this class, based on the biological intuition of allelic incompatibilities. 

In this Allelic Incompatibilities model, each allele has an additive fitness contribution. Beyond these additive contributions, for some pairs of loci, ``compatible'' and ``incompatible'' pairs of alleles are chosen. Each pair of ``incompatible'' alleles present in a genotype contributes negatively to its fitness.

Formally, if the allele at the $i$th locus in the landscape is denoted as $S_i=\pm 1$, the biallelic version of the model can be defined by the fitness function
\beq
f(g)=\frac{1}{2}\sum_{i=1}^Ls_iS_i+\sum_{i=1}^L\sum_{j=i+1}^LJ_{ij}S_iS_j
\eeq
where the $s_i$s describe the additive fitness effects, and the coefficients $J_{ij}$ describe the structure of epistatic interactions in the model. Such models are actually well known from the physics of magnetic systems and spin glasses - the Ising model and its generalisations like the Sherrington-Kirkpatrick model \cite{stein1992}.

If the $s_i$s and $J_{ij}$s are distributed as centered Gaussian variables with variances $\sigma_s^2$ and  $\sigma_J^2$ respectively, and if the interactions are dense enough - i.e. every locus interacts with a large fraction of the others - and random, then we can approximate the result with the one for the corresponding mean-field model, i.e. the Sherrington-Kirkpatrick model which is an unstructured model, obtaining
\begin{align}
\phi_{rs}&=\frac{2}{\pi}\arcsin\left(\frac{4\sigma_J^2}{\sigma_s^2+4(L-1)\sigma_J^2}\right) \\
\phi_{ss}&=1-\frac{2}{\pi}\left[\arcsin\left(1-\frac{8\sigma_J^2}{\sigma_s^2+4(L-1)\sigma_J^2}\right)+2\arcsin\left(\frac{4\sigma_J^2}{\sigma_s^2+4(L-1)\sigma_J^2}\right)\right]
\end{align}
so we have the same relations (\ref{relfrs}), (\ref{relfs}) and sign epistasis dominates. 

The interesting case is the one with sparse, structured interactions. Assume that each site interacts with $I$ other sites, with sparse interactions ($I\ll L$). For the case when each site interacts with many others ($I\gg 1$), we have: 
\begin{align}
\phi_{rs}&=\frac{I}{L-1}\frac{2}{\pi}
\arcsin\left(\frac{4\sigma_J^2}{\sigma_s^2+4I\sigma_J^2}\right) \\
\phi_{ss}&=\frac{I}{L-1}\left[1-\frac{2}{\pi}
\left[\arcsin\left(\frac{\sigma_s^2+4(I-2)\sigma_J^2}{\sigma_s^2+4I\sigma_J^2}\right)+2\arcsin\left(\frac{4\sigma_J^2}{\sigma_s^2+4I\sigma_J^2}\right)\right]\right]
\end{align}
while for the case of extremely sparse interactions ($I\ll 1$), we have
\begin{align}
\phi_{rs}&=\frac{I}{L-1}\frac{2}{\pi}
\arcsin\left(\frac{4\sigma_J^2}{\sigma_s^2+4\sigma_J^2}\right) \\
\phi_{ss}&=\frac{I}{L-1}\left[1-\frac{2}{\pi}
\left[\arcsin\left(\frac{\sigma_s^2-4\sigma_J^2}{\sigma_s^2+4\sigma_J^2}\right)+2\arcsin\left(\frac{4\sigma_J^2}{\sigma_s^2+4\sigma_J^2}\right)\right]\right].
\end{align}
For large landscapes, these sparse models show weak local epistasis with $1-\gamma=O(I/L)$. Most interestingly, in models with $I\leq 2$, reciprocal sign epistasis dominates when the interaction terms are strong enough compared to the linear terms, i.e. $4(2-I)\sigma_J^2>\sigma_s^2$ and $4\sigma_J^2>\sigma_s^2$ respectively. (Beware that epistasis in this case is not necessarily weak when measured using global epistasis measures, as expressed for example by roughness/slope statistics.)

Therefore, models with extremely sparse interactions involving allelic incompatibilities provide a counterexample to the trend of SSE dominance that we observed in all other models and in empirical landscapes. 



\subsubsection{Sign epistasis in a computational RNA landscapes}

In the previous section, we presented a theoretical counterexample to the predominance of SSE in weakly epistatic landscapes. Here we present a computational fitness landscape that represents a related, biologically motivated counterexample. 

As we have shown, epistatic blocks tend to generate comparable amounts of RSE and SSE, both at weak ($\gamma\approx 1$) and strong ($0<\gamma< 1$) epistasis, while allelic incompatibilities can generate landscapes dominated by RSE. These interaction patterns are naturally present in the folding of RNA sequences into their secondary structure. In fact, Watson-Crick pairings are typical examples of allelic incompatibilities, while the maintenance of the stem-loop structures induces important local interactions. In this section we perform a statistical analysis of a randomly chosen example of RNA folding landscape, to illustrate how the properties of such a landscape may differ significantly from the empirical landscapes studied before. In particular, SSE does not necessarily dominate over RSE in RNA landscapes.

We consider the secondary structure of the well-studied cricket paralysis virus IRES \cite{schuler2006structure}. Internal ribosome entry sites (IRESs) are RNA elements contributing to translation initiation in RNA viruses; the stability of their secondary and tertiary structure is paramount in ensuring the initiation of translation of capsid proteins, and hence the reproduction of the virus. 

This IRES is 317 bases long. We computed the free energy of all possible 450,774 double mutants from the wild type using the ViennaRNA package \cite{lorenz2011viennarna}. We considered the landscape defined by these free energy values on the genotype space of all the genotypes two mutations away from the wild type, and all mutations among those genotypes. This way, we avoid any bias towards the wild type, since all genotypes have the same Hamming distance from it. Then we computed $\phi_{ss}$ and $\phi_{rs}$ for different sub-landscapes with mutations restricted to a small number of sites. (Recall that neutral mutations do not count towards SSE or RSE.) 

The epistatic interactions in this IRES sequence are quite strong for a real landscape of this size, with a measure of epistasis $1-\gamma^*=0.56$. The amount of sign epistasis among random sites is $\phi_{ss}=20\%(\pm4\%)$ while reciprocal sign epistasis is only slightly lower at $\phi_{rs}=18\%(\pm3\%)$. Interactions have a quite strong local component, as shown in Figure \ref{fig:ires}: the epistasis for mutations a few sites away can be as high as $1-\gamma^*\approx 0.75$ vs $1-\gamma^*\approx 0.56$ for random sites. However, the average strength of epistasis is remarkably high even between distant sites.

Local interactions show a pattern of almost perfect balance between SSE and RSE (similarly to the expected outcome of a block model, or a model with neither too sparse nor too dense allelic incompatibilities). Even epistatic interactions among more distant sites show only a slight excess of SSE over RSE. Overall, these results support the idea that both block structures and allelic incompatibilities are relevant if one wants to understand and model interactions in such an example of RNA landscape.

\begin{figure}
\begin{center}
\includegraphics[width=15cm]{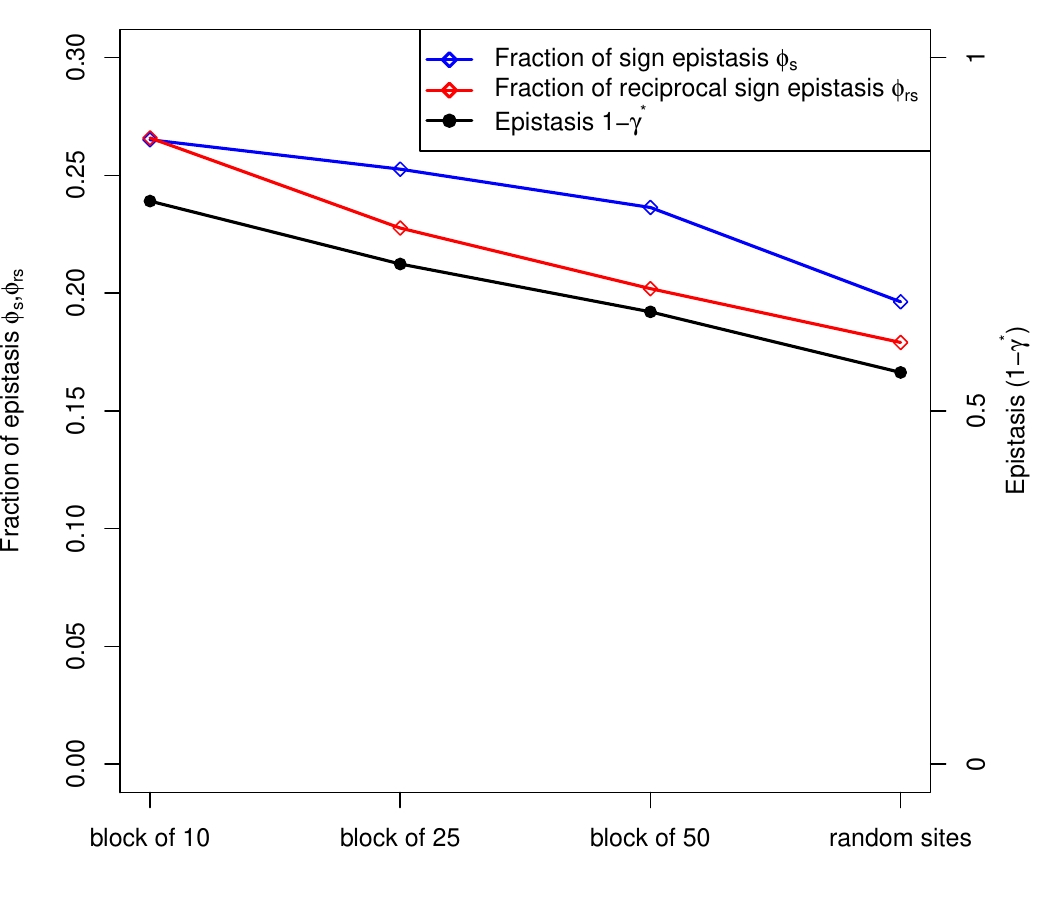}
\caption{\small Statistical analysis of the computational RNA landscape reconstructed by ViennaRNA for the cricket paralysis virus IRES. Local fitness graph epistasis ($1-\gamma^*$) and fractions of simple/reciprocal sign epistasis are averaged across randomly chosen sub-landscapes of this landscape; these sub-landscapes can be built from mutations in blocks of 10, 20 or 50 contiguous sites along the sequence, or from randomly chosen sites along the whole RNA molecule. We limit our analysis to the space of all genotypes two mutations away from the IRES wild type.}\label{fig:ires}
\end{center}
\end{figure}

\section{Discussion}

We showed how simple sign epistasis is associated with evolutionary detours. The importance of these evolutionary detours has been emphasised before both for empirical landscapes \cite{wu2016adaptation} and for theoretical landscapes
\cite{zagorski2016beyond,Kaznatcheev2019}. In fact, models of holey landscapes \cite{Gavrilets2004} has been developed to illustrate how high-dimensional epistatic landscapes can be navigable thanks to evolutionary detours. 
The role of SSE contrasts partly with the role of RSE, which is also correlated to other measures of evolutionary constraints, like the number of local peaks and, more generally, the number of genotypes with few beneficial mutations [Ferretti et al, in preparation] such as evolutionary chains \cite{ferretti2018evolutionary}. As we also showed, it is negatively related to the number of fitness-increasing paths to the absolute fitness maximum. 

The theorems presented here are relevant to understand the properties of real landscapes. Due to the exponentially large size of fitness landscapes, only a small fraction of real landscapes is amenable to experimental studies, hence it is important to find local correlates of global evolutionary properties. The existence and abundance of fitness peaks, evolutionary constraints and detours is difficult to assess experimentally, since it requires to study all mutations in large regions of the landscape. On the other hand, local quantities like the abundance of SSE and RSE can be estimated from data on small sub-landscapes or sparse pairs of mutations. 


We quantified the amount of sign epistasis, SSE, RSE and fitness graph epistasis $1-\gamma^*$ for both landscape models and empirical fitness landscapes. Overall sign epistasis, RSE and $1-\gamma^*$ increase monotonically with the local epistasis measure $1-\gamma$, both in landscape models and across empirical landscapes. SSE is also monotonic for reasonable values of epistasis but, on the other hand, shows non-monotonicity when extreme structured epistasis ($\gamma<0$) appears. This corresponds to the unlikely scenario where the fitness effects of mutations becomes anti-correlated after changes in the genetic background, and it is expected to occur only for some interactions of a few mutations. In fact, we observe it only in some sublandscapes of size 3.

The relation between SSE, RSE and other quantities takes a very simple form in unstructured Gaussian landscape models, i.e. isotropic Gaussian random field models. In all these models, despite their complexity, these quantities depend only on the correlation of fitness effects $\gamma$. For example, in this class of models the relation between the epistasis computed from the fitness graph and $\gamma$ is simply $\gamma^*\approx\frac{2}{\pi}\arcsin\gamma$.

These isotropic Gaussian random field models are particularly interesting because they have been recently used as natural priors for fitness landscape inference \cite{zhou2022higher,zhou2025learning}. Their properties of invariance/isotropy and their Gaussianity makes them natural candidates for the role. Even more interestingly, unstructured models appear also to provide an excellent fit to empirical landscapes in terms of SSE and RSE, further demonstrating how the choice of these landscapes as priors in appropriate for molecular landscapes. This also suggests that typical empirical fitness landscapes may be relatively unstructured in their interactions.

The exceptions are landscapes characterised by a strong block structure and by allelic incompatibilities, which are the only ones where RSE dominates. This result suggests that complex, strongly interacting structures like RNA molecules could  represent a significant obstacle to evolutionary innovations. In fact, single mutations do not seem to be effective in exploring fitness landscapes with 
a large amount of reciprocal sign epistasis. The effect should be stronger in organisms with a RNA genome, such as RNA viruses. Yet, the sequence of RNA viruses shows fast evolutionary rates and their evolution does not appear to be strongly constrainted. A likely answer lies in the high mutation rates of RNA viruses \cite{drake1993rates}. These high rates are often attributed to the error-prone RNA polymerase, i.e. a byproduct of  evolutionary tinkering. However, it is conceivable that the fitness landscape for RNA viruses is structurally different from the one for DNA viruses, and that high mutation rates could actually allow the viruses to produce enough double and multiple mutants to escape the local constraints generated by epistatic interactions. In fact, reciprocal sign epistasis does not represent a significant obstacle to evolution if double mutations are frequent enough to generate a fraction of genotypes that could cross the fitness ``valley" \cite{ghafari2019expected}. Instead of looking at high mutation rates in terms of genetic load only, their role in facilitating the exploration of larger parts of the rugged fitness landscape of RNA viruses should be reconsidered. As often suggested in the past, high mutation rates could be key for evolutionary innovation in viruses. However, it should be noted that there are other possible explanations for the high mutation rate in RNA viruses and more generally in viruses \cite{elena2005adaptive}.

For the first time, we have shed light on the general properties of simple sign epistasis. Our theoretical expectations for the role and frequency of simple sign epistasis and reciprocal sign epistasis are in good agreement with the empirical results from a collection of real-life landscapes. This supports the suggestion that these theoretical insights will generalise to many biologically relevant fitness landscapes. 

\section*{Data availability statement}
All data analysed in this study are derived from publicly available sources. The datasets for all sublandscapes, along with the code used for simulations and for the analysis of the ViennaRNA landscape, are available in the Zenodo repository [ADD].

\section*{Acknowledgments}
We acknowledge support from the European REA, Marie Skłodowska-Curie Actions, grant agreement no. 101131463 (SIMBAD). This work was funded by UK Research and Innovation (UKRI) under the UK government’s Horizon Europe funding guarantee (grant number EP/Y037375/1). MG is supported by a Wellcome Trust Early-Career Award (grant 309205/Z/24/Z). JK acknowledges support by Deutsche Forschungsgemeinschaft (DFG) within CRC 1310.

\printbibliography

\clearpage

\setcounter{page}{1}

\appendix
\setcounter{section}{19}

\section*{Supplementary Information:\\
Simple sign epistasis and evolutionary detours in fitness landscapes
}
\emph{Ribeca} et al.
\subsection{Theorems on fitness graphs}
\subsubsection{Sign epistasis}

\textbf{Theorem:} \emph{The presence of SSE is a necessary and sufficient condition for either of these:\\
(i) the existence of at least a pair of non-neighbour genotypes that are connected by exactly one direct accessible path;
\\
(ii) the existence of at least one reversing path in the landscape.} 

\textbf{Proof:} To show that (i) and (ii) are sufficient conditions, just notice that both are easily verified within the SSE motif itself, as discussed in the Main Text: \\
(i) The genotypes connected by a single direct accessible path are the peak and its antipode, and there cannot be any other direct accessible path since all direct paths between these genotypes (at Hamming distance 2) belong to the same motif. \\
(ii) The reversing path is the path of length three that ends on the peak. \\
We now show that the two conditions are also necessary:
\\
(i) Consider the pair of genotypes connected by a single direct fitness-increasing path, and exchange any two consecutive mutations along this direct path $g\stackrel{\mu}{\rightarrow} g'\stackrel{\mu'}{\rightarrow} g''$ to obtain $g\stackrel{\mu'}{\rightarrow} g'''\stackrel{\mu}{\rightarrow} g''$. We have $f(g)<f(g')<f(g'')$ by construction, while the uniqueness of the path implies that either $f(g''')<f(g)$ or $f(g''')>f(g'')$. Hence the resulting motif formed by the four genotypes $g,g',g'',g'''$ is SSE. 
\\ 
(ii) Denote by $g_s,g_e$ the two genotypes connected by a reversing path, with $f(g_s)<f(g_e)$. Denote by $\mathcal{D}(g_s,g_e)$ the set of genotypes that lie along direct paths from $g_s$ to $g_e$. We can define a distance between a genotype and this set by $d_{g,\mathcal{D}(g_s,g_e)}=min_{x\in \mathcal{D}(g_s,g_e)}d_{x,g}$ where $d_{x,g}$ is the Hamming distance. We can then define the distance between any path $P$ from $g_s$ to $g_e$ and $\mathcal{D}(g_s,g_e)$ by $d_{P,\mathcal{D}(g_s,g_e)}=\sum_{g\in P} d_{g,\mathcal{D}(g_s,g_e)}$. Consider the reversing (fitness-increasing) path $\tilde{P}$ from $g_s$ to $g_e$ that is closest to $\mathcal{D}(g_s,g_e)$. Denote by $g'$ the last farthest genotype from $\mathcal{D}(g_s,g_e)$ on this reversing path $\tilde{P}$ and denote the mutations before and after $g'$ as $g\stackrel{\mu}{\rightarrow} g'\stackrel{\mu'}{\rightarrow} g''$. Exchanging these mutations we obtain $g\stackrel{\mu'}{\rightarrow} g'''\stackrel{\mu}{\rightarrow} g''$. Note that the mutation $\mu'$ reduces the distance from $\mathcal{D}(g_s,g_e)$ (while  $\mu$ does not), otherwise $g'$ would not be the last farthest genotype. Hence, $d_{g''',\mathcal{D}(g_s,g_e)}<d_{g',\mathcal{D}(g_s,g_e)}$. This implies that exchanging the mutations, we obtain a path that is closer to $\mathcal{D}(g_s,g_e)$, hence it cannot be fitness-increasing. All steps are fitness-increasing except possibly $g\stackrel{\mu'}{\rightarrow} g'''$ and $g'''\stackrel{\mu}{\rightarrow} g''$, and the combination of the two is beneficial, hence only one of the two should be deleterious. Since in the other order both $\mu$ and $\mu'$ were beneficial, the motif formed by the four genotypes $g,g',g'',g'''$ is SSE.

\subsubsection{Reciprocal sign epistasis}

\textbf{Theorem:} \emph{If there is at least a genotype that has no direct fitness-increasing path to a peak in the landscape, then there is RSE.}

\textbf{Proof:} Consider a 
direct path from the peak to genotype, chosen in such a way that at every step the fitness of the next genotype is maximised. The first mutation is fitness-decreasing since it starts from a peak, but by hypothesis there should be a fitness-increasing mutation along the path. Consider the first such fitness-increasing mutation $g'\stackrel{\mu'}{\rightarrow} g''$ and the previous one $g\stackrel{\mu}{\rightarrow} g'$ along the path and exchange them to obtain $g\stackrel{\mu'}{\rightarrow} g'''\stackrel{\mu}{\rightarrow} g''$. We have $f(g),f(g'')>f(g')$ by choice of mutations and $f(g')>f(g''')$ by choice of path, hence $f(g),f(g'')>f(g''')$ as well and the four genotypes $g,g',g'',g'''$ form a RSE motif.

Note that it is trivial to use Poelwik's theorem to prove that the reverse is also true, i.e. \emph{if there is RSE, then there is at least a genotype that has no direct fitness-increasing path to a peak in the landscape.} In fact, the second peak is precisely such a genotype. \\
This theorem was independently obtained in \cite{Kaznatcheev2019}.

\subsection{Frequency of types of epistasis: general approach}

Here we present the general approach used for all computations. 
For brevity, let us define 
\[ \Delta_i f(g) \equiv f(g_{[i]}) - f(g) \]
the change in fitness when allele $i$ is mutated in the biallelic case, and a more general expression $\Delta_{i,a} f(g)$ for the multiallelic case specifying the new value $a$ of the allele. 
Our results are based on the following proposition.

\textbf{Theorem 3}: \emph{extracting a random genotype $g$ from a landscape with fitness function $f(g)$ and two random mutations in different loci $i\neq j$, the following relations hold for the probabilities}
\beq
P\left[\Delta_{[i,a]} f(g) > 0,\Delta_{[j,a']} f(g) > 0\right]=\frac{1+\phi_{rs}}{4}\label{th_eq1}
\eeq
\beq
P\left[\Delta_{[i,a]} f(g) > 0,\Delta_{[i,a]} f(g_{[j,a']}) < 0\right]=\frac{\phi_{ss}+2\phi_{rs}}{4}\label{th_eq2}.
\eeq

From this theorem, we can easily obtain the results for the expected values over realisations of the landscape:
\begin{align}
\ev[\phi_{rs}]&=4 P\left[\Delta_{i,a} f(g) > 0,\Delta_{j,a'} f(g) > 0\right]-1 \label{evrs}\\
\ev[\phi_{ss}]&=4 P\left[\Delta_{i,a} f(g) > 0,\Delta_{i,a} f(g_{[j,a']}) < 0\right]-2\ev[\phi_{rs}] = \label{evs}\\
& = 2-8 P\left[\Delta_{i,a} f(g) > 0,\Delta_{j,a'} f(g) > 0\right]+4P\left[\Delta_{i,a} f(g) > 0,\Delta_{i,a} f(g_{[j,a']}) < 0\right] \nonumber
\end{align}
where the probabilities now refer to a double random sampling, i.e. both a random genotype and a random realisation of the landscape.

For Gaussian landscapes, we will use also the following result.

\textbf{Theorem 4}: \emph{for two Gaussian variables $\xi_1,\xi_2$ with $\ev[\xi_1]=\ev[\xi_2]=0$, $\var[\xi_1]=\var[\xi_2]=\sigma^2$ and $\cov[\xi_1,\xi_2]=c$, we have}
\beq
P[\xi_1>0,\xi_2>0]=\frac{1}{4}+\frac{1}{2\pi}\arcsin\left(\frac{c}{\sigma^2}\right).
\eeq

This expression depends only on the correlation between the Gaussian variables. This theorem will be used to compute the probabilities in equations (\ref{evrs}),(\ref{evs}) by choosing the right set of Gaussian variables. More in detail, the variables $\xi_1,\xi_2$ will correspond to $f(g_{[i]})-f(g),f(g_{[j]})-f(g)$ when computing equation (\ref{evrs}) and to $f(g_{[i]})-f(g),f(g_{[j]})-f(g_{[ij]})$ when computing equation (\ref{evs}).

Finally, a simple theorem that is useful in several computations related to Theorem 4, as well as in computing $\gamma$ for many landscapes models presented here.

\textbf{Theorem 5}:
\emph{for any landscape, the following relation holds between covariances:}
\begin{align}
\cov\left[\Delta_{[i,a]} f(g) ,\Delta_{[i,a]} f(g_{[j,a']}) \right]&=\var\left[\Delta_{[i,a]} f(g) \right]-2\cov\left[\Delta_{[i,a]} f(g) ,\Delta_{[i,a]} f(g_{[j,a']}) \right]
\end{align}
\emph{and therefore}
\begin{align}
\gamma&=\frac{\cov\left[\Delta_{[i,a]} f(g) ,\Delta_{[i,a]} f(g_{[j,a']}) \right]}{\var\left[\Delta_{[i,a]} f(g) \right]}=1-2\frac{\cov\left[\Delta_{[i,a]} f(g) ,\Delta_{[i,a]} f(g_{[j,a']}) \right]}{\var\left[\Delta_{[i,a]} f(g) \right]}.
\end{align}

\subsubsection{Proofs}
\textbf{Theorem 3}: this is a purely combinatorial theorem. The random extraction of a random genotype and two random mutations $i,j$ at different sites is equivalent to the random extraction of a motif from all motifs in the landscape, and then a random vertex from that motif. Therefore, the probabilities can be computed for each of the motifs ME, SSE, RSE, and then weighted by their relative abundance $1-\phi_{ss}-\phi_{rs},\phi_{ss},\phi_{rs}$. The probability that a random vertex is a minimum within the motif is $1/4$ for ME and SSE, and 1/2 for RSE. Weighting by the abundance of different motifs, we obtain the result (\ref{th_eq1}). Similarly, the probability that the fitness effect of a random mutation in the motif changes sign is $0$ for ME, $1/2$ for SSE and $1$ for RSE; when we also require that the mutation is initially beneficial, these probabilities are halved to $0$ for ME, $1/4$ for SSE and $1/2$ for RSE. We obtain (\ref{th_eq2}) by weighting by the abundance of different motifs.\\
\textbf{Theorem 4}: the probability is equivalent to the integral of the joint distribution $\xi_1,\xi_2$ over the positive quadrant of $\mathbb{R}^2$, i.e. the subspace defined by the condition $(\xi_1>0,\xi_2>0)$, or equivalently the slice between the two vectors $(1,0),(0,1)$. After a change of variables to $x_\pm=\frac{\xi_1\pm\xi_2}{\sqrt{2(1+c/\sigma^2)}}$, these new variables have covariances given by $\var(x_+)=\var(x_-)=1$ and $\cov(x_+,x_-)=0$. In the $(x_+,x_-)$ plane, the slice is still centered at $(0,0)$ but now defined as the slice between the two transformed vectors
\[
  \left(\frac{1}{\sqrt{2(1+c/\sigma^2)}},\frac{1}{\sqrt{2(1-c/\sigma^2)}}\right)\;
  \left(\frac{1}{\sqrt{2(1+c/\sigma^2)}},-\frac{1}{\sqrt{2(1-c/\sigma^2)}}\right).
\]
The distribution is still Gaussian, but now spherically symmetric in this space, therefore the integral reduces to the fraction of the total angle $2\pi$ covered by the slice. The cosine of this angle is the scalar product between the vectors divided by their length, which results in $-c/\sigma^2$, therefore the angle is $\theta=\arccos(-c/\sigma^2)$. Since $\cos(\theta)=\sin(\pi/2-\theta)$, the fraction of the total angle can be written as
\begin{align}
P[\xi_1>0,\xi_2>0]=\frac{\arccos(-c/\sigma^2)}{2\pi}=\frac{1}{4}+\frac{1}{2\pi}\arcsin\left(\frac{c}{\sigma^2}\right).
\end{align}

\subsection{Frequency of types of epistasis: specific models}

We assume that a genotype is given by a vector $g=[x_1,x_2,\ldots x_L]$ of alleles $x_i \in \{\pm1\}$ (in the biallelic case) or $x_i \in \{a_1\ldots a_A\}$ (in the multiallelic case). The fitness function $f(g)$ is a real function defined over the space of all $A^L$ different genotypes. In the biallelic case, we use the notation $g_{[i]} = [x_1,x_2,\ldots,-x_i,\ldots, x_L]$ to denote the genotype that differs from $g$ only in the allele $x_i$ at locus $i$. We also denote the change in fitness $\Delta_i f(g)=f(g_{[i]})-f(g)$.

\subsubsection{Unstructured Gaussian (isotropic random field) models}
This family of models is characterised by \(f(g) \sim N(\bar{f},\Sigma^2)\) where $\bar{f}$ is a constant and the covariance matrix $\Sigma^2$ has the following symmetries: invariance under arbitrary permutations of loci
\[\mathrm{Cov}[f(g),f(g')]=\mathrm{Cov}[f(g_{x_i\leftrightarrow x_j}),f(g'_{x_i\leftrightarrow x_j})]\]
and invariance under arbitrary permutations of allele at a given locus
\[\mathrm{Cov}[f(g),f(g')]=\mathrm{Cov}[f(g_{x_i=(a\leftrightarrow a')}),f(g'_{x_i=(a\leftrightarrow a')})]\]

These properties make it simple to compute the probabilities using Theorem 4. In fact, $\Delta_{i,a} f(g) = f(g_{[i,a]})-f(g), \Delta_{j,a'} f(g) =f(g_{[j,a']})-f(g)$ is a pair of centered Gaussian random variables with variance equal to $\mathrm{E}[(\Delta f)^2]=2\mathrm{Var}[f](1-\rho_1)$ and covariance $\mathrm{Var}[f](\rho_2-2\rho_1+1)$ where $\rho_d$ are the fitness correlation functions at distance $d$ for this landscape. Similarly, $\Delta_{i,a} f(g),\Delta_{i,a} f(g_{[j,a']})$ is a pair of centered Gaussian random variables with variance equal to $2\mathrm{Var}[f](1-\rho_1)$ and covariance $2\mathrm{Var}[f](\rho_1-\rho_2)$. Since the correlation of fitness effects $\gamma$ is defined as $\gamma=\frac{\rho_1-\rho_2}{1-\rho_1}$, we can rewrite the correlations
\[
\mathrm{Cor}[\Delta_{i,a} f(g) ,\Delta_{i,a} f(g_{[j,a']})]=\gamma
\]
\[
\mathrm{Cor}[\Delta_{i,a} f(g) ,\Delta_{j,a'} f(g) ]=\frac{1-\gamma}{2}
\]
and use Theorem 4 and then Theorem 3 to obtain the final results for SSE and RSE.

\subsubsection{Rough Mount Fuji model}

\textbf{Classical version:} In this biallelic landscape, each locus contributes with a fixed ``slope'' $s>0$, and the roughness is given by a random House-of-Cards contribution:
\[
f(g) = \frac{s}{2} \sum_i x_i + \eta(g)
\]
where $\eta(g)$ are i.i.d. random variables extracted for each genotype. Here we consider the choice of random uniform variables \[\eta \sim U[0,\sigma].\]
This is the ``classical'' version because the additive effects are deterministic constants, not drawn from a distribution.

To derive the required probabilities in this model, first we scale all fitness values by $\sigma$ so that $\eta$ is $U[0,1]$-distributed and the only parameter of the model is the ratio $s'=s/\sigma$. Then, we have to solve the integrals
\begin{align*}
P[\Delta_i f(g) >0,\Delta_i f(g_{[j]})<0] & = \int_0^1 dx \int_0^1 dy \int_0^1 dz \int_0^1 dw\ \theta(x-z-s')\theta(y-w+s')\\
P[\Delta_i f(g)>0,\Delta_i f(g_{[j]})>0] & = \frac{1}{4}\sum_{\chi,\xi\in\{\pm 1\}}\int_0^1 dx \int_0^1 dy \int_0^1 dz\ \theta(x-z+\chi s')\theta(y-z+\xi s')
\end{align*}


The first integral factorises as 
\begin{align*}
P[\Delta_i f(g) >0,\Delta_i f(g_{[j]})<0] &=
\int_0^1 dx \int_0^1 dy\ \theta(x-y-s') \left[1-\int_0^1 dx \int_0^1 dy\ \theta(x-y-s')\right] \\
&= \frac{(1-s')^2}{2}\left(1-\frac{(1-s')^2}{2}\right)\theta(1-s')
\end{align*}
where $\theta(x)$ is the Heaviside step function. The second equation contains three different integrals, which evaluate to 
\begin{align*}
\int_0^1 dx \int_0^1 dy \int_0^1 dz\ \theta(x-z- s')\theta(y-z- s') & = 1+\left(s'-1+\frac{1-s'^3}{3}\right)\theta(1-s')\\
\int_0^1 dx \int_0^1 dy \int_0^1 dz\ \theta(x-z+ s')\theta(y-z+ s') & = \frac{(1-s')^3}{3} \theta(1-s')\\
\int_0^1 dx \int_0^1 dy \int_0^1 dz\ \theta(x-z- s')\theta(y-z+ s') & = \theta(1-s')\left[\frac{(1-s')^2}{2} + \frac{(1-2s')^3}{6}\theta(1/2-s') \right].
\end{align*}
Using Theorem 3, we finally find the solution
\begin{align}
\phi_{rs}&=\begin{cases} 0 & \mathrm{for}\ \frac{s}{\sigma}\geq {1}\\
\frac{2}{3}\left(1-\frac{s}{\sigma}\right)^3
& \mathrm{for}\ 1>\frac{s}{\sigma}>\frac{1}{2}\\ 
\frac{2}{3}\left(1-\frac{s}{\sigma}\right)^3-\frac{1}{3}\left(1-2\frac{s}{\sigma}\right)^3
& \mathrm{for}\ \frac{s}{\sigma}\leq \frac{1}{2}\end{cases}\\
\phi_{ss}&=\begin{cases} 0 & \mathrm{for}\ \frac{s}{\sigma}\geq {1}\\
1- \left(\frac{s}{\sigma}\right)^2\left(2-\frac{s}{\sigma}\right)^2
-\frac{4}{3}\left(1-\frac{s}{\sigma}\right)^3
& \mathrm{for}\ 1>\frac{s}{\sigma}>\frac{1}{2}\\ 
1- \left(\frac{s}{\sigma}\right)^2\left(2-\frac{s}{\sigma}\right)^2
-\frac{4}{3}\left(1-\frac{s}{\sigma}\right)^3+\frac{2}{3}\left(1-2\frac{s}{\sigma}\right)^3
& \mathrm{for}\ \frac{s}{\sigma}\leq \frac{1}{2}\end{cases}
\end{align}
or equivalently, the one reported in the Main Text.

{\bf Unstructured version:}
It is a variation of the classical model, but allowing for a random orientation of the peak:
\[
f(g) = \sum_{i=1}^L \frac{a_i}{2} x_i + \eta(g)
\]
where  \(a_i\) are random variables (that do not depend on the genotype), while $\eta$ are i.i.d. random variables extracted for each genotype. We consider the Gaussian case where the slope vector and noise are both Gaussians, with parameters
\[ \vec a\: :\: a_i  \sim N(\mu_a,\sigma_a^2) \quad \eta  \sim N(0,\sigma_{HoC}^2).\]
The covariance structure for this model would be:
\[
\operatorname{Var}[\Delta_i f] = a_i^2 + 2\sigma_{\text{HoC}}^2, \quad \operatorname{Cov}[\Delta_i f, \Delta_j f] = \sigma_{\text{HoC}}^2
\]
with a third covariance involving the specific \(a\) values. 
Averaging over realisations, we obtain
\begin{align*}
\var\left[\Delta_i f(g) \right]=&\sigma_s^2+2\sigma_{HoC}^2\\
\cov\left[\Delta_i f(g) ,\Delta_j f(g) \right]=&\sigma_{HoC}^2\\
\cov\left[\Delta_i f(g) ,\Delta_i f(g_{[j]}) \right]=&\sigma_s^2.
\end{align*}
The typical correlation in fitness effects is given by \cite{Ferretti2016}:
\[ E[\gamma] = 1 - \frac{ 2 \sigma^2_{HoC} }{ \mu^2_a + \sigma^2_a + 2 \sigma^2_{HoC}}. \]
The unstructured case in the Main Text corresponds to $\mu_a=0$, since in this case the model is an unstructured Gaussian model, with epistasis 
\[ 1-\gamma =  \frac{ 2\sigma^2_{HoC} }{ \sigma^2_a + 2 \sigma^2_{HoC}}. \]

\subsubsection{House of Cards model}

\textbf{House of Cards (HoC)} It can be seen as a particular case of the Rough Mount Fuji model, where there is actually no general slope $s=0$, or $\vec a \equiv 0$, and all that remains is the rough random landscape (equivalent to a Random Energy Model):
\[
f(g) = \eta(g), \quad \eta(g) \text{ i.i.d.}
\]
This gives \(\operatorname{Var}[\Delta_i f] = 2\sigma_{\text{HoC}}^2\) and \(\operatorname{Cov}[\Delta_i f, \Delta_j f] = \sigma_{\text{HoC}}^2\) because mutations share the target genotype's random component.

\subsubsection{NK models}
\textbf{Mean-field NK model:} It is the mean-field version of the random NK model. In the random NK model, each locus interacts with \(K\) randomly chosen other loci.
\[
f(g) = \sum_{i=1}^L f_i(x_i, x_{i_1}, \dots, x_{i_{K}})
\]
where each \(f_i\) is a random table of size \(2^{K+1}\). The ``random'' refers to the random, non-geometric choice of interacting partners. 

In the mean-field model, all possible groups of loci of size $K+1$ interact: 
\[
f(g) \propto \sum_{\mathbf{J}\in \mathcal{S}_{K+1}} f_\mathbf{J}(x_{J_1}, \dots, x_{J_{K+1}})
\]
where $\mathcal{S}_{K+1}$ is the space of all subsets of loci of size $K+1$.

Here we consider all random variables in this model to be Gaussian-distributed with variance $\sigma_{NK}^2$. The correlations are (up to an irrelevant multiplicative constant):
\begin{align}
\var\left[\Delta_i f(g) \right]=&2(K+1)\sigma_{NK}^2\\
\cov\left[\Delta_i f(g) ,\Delta_j f(g) \right]=&\frac{K(K+1)}{L-1}\sigma_{NK}^2\\
\cov\left[\Delta_i f(g) ,\Delta_i f(g_{[j]}) \right]=&2(K+1)\left(1-\frac{K}{L-1}\right)\sigma_{NK}^2.
\end{align}

\textbf{NK model with adjacent interactions:} Loci arranged on a circle or line. Each locus \(i\) interacts with itself and the nearest \(K\) loci:
\[
f(g) = \sum_{i=1}^L f_i(g_i, g_{i+1}, \dots, g_{i+K})
\]
where indices mod \(L\). ``Adjacent'' means nearest-neighbor on the genome, not random.

In this case, correlations depend on distance $d$ between different loci. This dependence is explicitly related to the overlap of interaction neighborhoods between two loci $i$ and $j$. The maximum overlap will be $K+1$ and the term \(\min(d, K+1, L-K-1)\) encapsulates all possible configurations of $i$ and $j$ in the topology. 

Once more, assuming that all random variables are considered to be Gaussian-distributed with variance $\sigma_{NK}^2$, we obtain 
\begin{align}
\var\left[\Delta_i f(g) \right]=&2(K+1)\sigma_{NK}^2\\
\cov\left[\Delta_i f(g) ,\Delta_j f(g)\right]=&(K+1-\min(d,K+1,L-K-1))\sigma_{NK}^2\\
\cov\left[\Delta_i f(g) ,\Delta_i f(g_{[j]})\right]=&2\min(d,K+1,L-K-1)\sigma_{NK}^2.
\end{align}

\subsubsection{Block model}

This is also a NK model. The entire genotype is partitioned in $B$ blocks $g = [g_1,\ldots,g_{B}] = [\vec x_1, \vec x_2,\ldots, \vec x_{B}]$, each block containing $K+1=L/B$ alleles. Each block is independent (no cross-block interaction), therefore
\[ f(g) = \sum_{b\in 1}^B f_b(g_b)\]
where each \(f_b(g_b = \vec x_b)\) is a House of Cards. 

Since the chance that a mutations interact with another is $K/(L-1)$, the same factor rescales all epistatic statistics, i.e. 
\[\phi_{ss}=\frac{K}{L-1}\phi_{ss}^{(HoC)}=\frac{1/B-1/L}{1-1/L}\frac{1}{3}\]
\[\phi_{rs}=\frac{K}{L-1}\phi_{rs}^{(HoC)}=\frac{1/B-1/L}{1-1/L}\frac{1}{3}\]
\[1-\gamma=\frac{K}{L-1}(1=\gamma^{(HoC)})=\frac{1/B-1/L}{1-1/L}.\]

\subsubsection{Pairwise allelic incompatibilities model}

Pairwise models are those more closely related to a spin-glass / Ising model. For these biallelic models the genotypes are $g=(S_1\ldots S_L)$ with $S_i=\pm 1$, and the fitness function is similar to standard models in statistical mechanics:
\[
f(g) = \frac{1}{2}\sum_i s_i S_i +\sum_{i<j} J_{ij} S_i S_j
\]
where \(s_i\) are random fields, and \(J_{ij}\) are i.i.d. random couplings with mean 0 and variance \(\sigma_J^2\), assumed Gaussian here. ``Dense'' means that many of the \(\binom{L}{2}\) possible interactions are present, while in the ``sparse'' version, the only nonzero interactions $J_{ij}\neq0$ belong to a given sparse graph $(i,j) \in G$ with degree $I$.

The model with dense, random interactions can be approximated by a ``mean-field'' unstructured model with all interactions (i.e. $I=L-1$):
\begin{align}
\var\left[\Delta_i f(g) \right]=&\sigma_s^2+4(L-1)\sigma_J^2\\
\cov\left[\Delta_i f(g) ,\Delta_j f(g) \right]=&4\sigma_J^2\\
\cov\left[\Delta_i f(g) ,\Delta_i f(g_{[j]}) \right]=&\sigma_s^2+4(L-3)\sigma_J^2.
\end{align}

Instead, in the sparse limit $I\ll L$, the SSE and RSE for models with sparse, structured interactions can be reconstructed by noticing that the probabilities (\ref{th_eq1}),(\ref{th_eq2}) take the trivial values 1/4 and 0 respectively for a fraction $(L-1-I)/(L-1)$ of all possible pairs of mutations, and they take non-trivial values only for a fraction $I/(L-1)$ of pairs of mutations. For these interacting pairs of mutations, the variables that appear in these probabilities are distributed as multivariate Gaussians, and therefore Theorem 4 can be applied using the following covariances if $I\geq 1$:
\begin{align}
\var\left[\Delta_i f(g) \right]=&\sigma_s^2+4I\sigma_J^2\\
\cov\left[\Delta_i f(g) ,\Delta_j f(g) \right]=&4\sigma_J^2\\
\cov\left[\Delta_i f(g) ,\Delta_i f(g_{[j]}) \right] = &\sigma_s^2+4(I-2)\sigma_J^2
\end{align}
or these covariances for the extremely sparse case $I\ll 1$, which are similarly derived by conditioning on the existence of an interaction between these mutations:
\begin{align}
\var\left[\Delta_i f(g) \right]=&\sigma_s^2+4\sigma_J^2\\
\cov\left[\Delta_i f(g) ,\Delta_j f(g) \right]=&4\sigma_J^2\\
\cov\left[\Delta_i f(g) ,\Delta_i f(g_{[j]}) \right] = &\sigma_s^2-4\sigma_J^2.
\end{align}

\end{document}